\documentclass[11pt]{llncs}

\newlength{\parindentsave}

\begin{document}

\begin{frontmatter}

\title{An $O^*(1.1939^n)$ time algorithm for minimum weighted dominating induced matching}

\author{
Min Chih Lin \thanks{Partially supported by UBACyT Grants 20020100100754 and 20020090100149, PICT ANPCyT Grant 1970 and PIP CONICET Grant 11220100100310.} \inst{1},
Michel J. Mizrahi$^\star$ \inst{1} \and \\
Jayme L. Szwarcfiter \thanks{Partially supported by CNPq, CAPES and FAPERJ, research agencies. Presently visiting the Instituto Nacional de Metrologia, Qualidade e Tecnologia, Brazil.} \inst{2}
}

\institute{CONICET, Instituto de C\'alculo and Departamento de Computaci\'on\\ Universidad de Buenos Aires\\ Buenos Aires, Argentina
 \and
Inst de Matem\'atica, COPPE and NCE  \\ Universidade Federal do Rio de Janeiro \\
Rio de Janeiro, Brazil \\
\email{oscarlin@dc.uba.ar}, \email{michel.mizrahi@gmail.com}, \email{jayme@nce.ufrj.br}}

\maketitle              % typeset the title of the contribution

\begin{abstract}
Say that an edge of a graph $G$ dominates itself and every other
edge adjacent to it. An edge dominating set of a graph $G=(V,E)$
is a subset of edges $E' \subseteq E$ which dominates all edges of
$G$. In particular, if every edge of $G$ is dominated by exactly
one edge of $E'$ then $E'$ is a dominating induced matching. It is
known that not every graph admits a dominating induced matching,
while the problem to decide if it does admit it is NP-complete. In
this paper we consider the problems of finding a minimum weighted
dominating induced matching, if any, and counting the number of dominating induced 
matchings of a graph with weighted edges. We describe an exact algorithm for 
general graphs that runs in $O^*(1.1939^n)$ time and polynomial (linear) space. 
This improves over any existing exact algorithm for the problems in consideration.
\end{abstract}

\keywordname{exact algorithms, dominating induced matchings, branch $\&$ reduce}

\end{frontmatter}

\pagestyle{plain}
\setcounter{page}{1}
\pagenumbering{arabic}
\section{Introduction}

Under the widely accepted assumption that $P \ne NP$ there are several problems with important applications for which no polynomial 
algorithm exists. 
The need to get an exact solution for many of those problems has lead to a growing interest in the area of design and analysis of 
exact exponential time algorithms for NP-Hard problems \cite{Fo-Kr,Woeginger:exactalgorithms}. Even a slight improvement of the base of the exponential running time
may increase the size of the instances being tractable. 
There has been many new and promising advances in recent years towards this direction \cite{Bjorklund2010,Bjorklund2007}.

In this paper we give an exact algorithm for the weighted and counting version of the NP-Hard problem Dominating Induced Matching
(also known as DIM or Efficient Edge Domination) which has been extensively studied 
\cite{Ca-Ce-De-Si,Do-Lo,Korpel,Lu-Ko-Ta,Lu-Ta,Br-Le-Ra,Br-Mo,Br-Hu-Ne,Ca-Ko-Lo}. Further notes about this problem and some applications 
related to encoding theory, network routing and resource allocation can be found in \cite{Gr-Sl-Sh-Ho,Livingston}.

The unweighted version of the dominating induced matching
problem is known to be NP-complete \cite{Gr-Sl-Sh-Ho}, even for
planar bipartite graphs of maximum degree 3 \cite{Br-Hu-Ne} or regular graphs
\cite{Ca-Ce-De-Si}. There are polynomial time algorithms for some
classes, such as chordal graphs \cite{Lu-Ko-Ta}, generalized
series-parallel graphs \cite{Lu-Ko-Ta} (both for the weighted
problem), claw-free graphs \cite{Ca-Ko-Lo}, graphs with bounded
clique-width \cite{Ca-Ko-Lo}, hole-free graphs \cite{Br-Hu-Ne}, 
convex graphs \cite{Korpel}, dually-chordal graphs \cite{Br-Le-Ra}, 
$P_7$-free graphs \cite{Br-Mo}, bipartite permutation graphs \cite{Lu-Ta} (see also \cite{Br-Lo}).

If $P \ne NP$ it is not possible to solve this problem in polynomial time, hence it becomes
important to improve the exponential algorithm in order to identify the instances that can be solved 
within reasonable time. 

A straightforward brute-force algorithm to solve weighted DIM can be achieved in $O^*(2^n)$ time and polynomial space.

The minimum weighted DIM problem can be expressed as an instance
of the maximum weighted independent set problem on the square of
the line graph $L(G)$ of $G$, and also as an instance of the
minimum weighted dominating set problem on $L(G)$, by slightly
way described in \cite{Br-Le-Ra,Milanic} for unweighted DIM problem.

The minimum weighted dominating set can be solved in $O^*(1.5780^n)$ time \cite{Fomin05boundingthe}, while the maximum weight independent set can be 
solved in $O^*(1.4423^n)$ time by enumeration of all maximal independent sets \cite{Ts}. To the best of our knowledge there are no better method to obtain
the maximum weighted independent set (a better algorithm $O^*(1.2209^n)$ for unweighted  maximum independent set is due by \cite{Fo-Gr-Kr}).
Hence the DIM problem for a graph $G$ can be solved by using this algorithm in $L^2(G)$, which runs in $O^*(1.4423^m)$ time.

For the minimum weighted DIM this algorithm behaves better than the brute-force alternative whenever $1.4423^m < 2^n$, this is, $m < 1.8926n$.  

The paper \cite{DBLPMin} shows how to solve the DIM problem in $O^*(1.7818^n)$ time and polynomial space
while the same algorithm runs in $O(n+m)$ time if the graph has a fixed dominating set. 
In the same work another approach based on enumerating maximal independent sets was developed and allows to solve both DIM problems 
(minimum weighted problem and counting problem) in $O^*(1.4423^n)$ time and polynomial space.
For the counting problem, there exists algorithms such as \cite{Dahllof2002} which can be used to count the number
of MWIS's in $O^*(1.3247^n)$ time, leading an $O^*(1.3247^m)$ time and polynomial space algorithm to count the numbers of DIM's. 

Comparing with the straightforward brute-force algorithm, it is convenient to use it as long as $1.3247^m < 2^n$, and this occurs whenever $m < 2.4650n$. 

There are NP-complete instances of the DIM problem where the number of edges in G is relatively low such as for planar bipartite graphs of maximum degree 3 \cite{Br-Hu-Ne}, where $m \leq 1.5n$. Therefore using transformations and exact algorithms for MWIS or for counting MWIS's is better than using the brute-force
algorithm. Note however that cases where brute-force algorithm is not convenient strongly relies on the number edges. For instance, for a graph with $O(n)$ edges
such that $m \approx 3n$ the MWIS algorithm behaves better than the brute-force one.

In this paper, we propose an algorithm for solving the weighted
DIM problem and the counting DIM's problem, in $O(m \cdot 1.1939^n) \in O^*(1.1939^n)$ time and $O(m)$ space in general graphs 
which improves over the existing algorithms for these problems. We employ techniques described in \cite{Fo-Kr} for the analysis of our algorithm,
and as such we use their terminology. 

The proposed algorithm was designed using the {\it branching $\&$ reduce} paradigm. More information about this design technique as well as the running time analysis
for this kind of algorithms can be found in \cite{Fo-Kr}. 

\section{Preliminaries}

By $G(V,E)$ we denote a \emph{simple undirected graph} with vertex
set $V$ and edge set $E$, $n=|V|$ and $m=|E|$. We consider $G$ as
a \emph{weighted} graph, that is, one in which there is a
non-negative real value, denoted $weight(vw)$ assigned to each
edge $vw$ of $G$. If $v \in V$ and $V' \subseteq V$, then denote
by $N(v)$, the set of vertices adjacent (neighbors) to $v$, 
denote $d(v) = |N(v)|$ the degree of the vertex, denote
by $G[V']$ the subgraph of $G$ induced by $V'$, and write
$N_{V'}(v)=N(v) \cap V'$. Some special graphs or vertices are of
interest for our purposes. A graph formed by two triangles having
a common edge is called a {\it diamond}. By
removing an edge incident to a  vertex of degree 2 of a diamond,
we obtain a {\it paw}. Finally, a vertex of degree 1 is called
{\it pendant}.

Given an edge $e \in E$, say that $e$ \emph{dominates} itself and
every edge sharing a vertex with $e$. Subset $E'\subseteq E$ is an
\emph{induced matching} of $G$ if each edge of $G$ is dominated by
at most one edge in $E'$. A \emph{dominating induced matching
(DIM)} of $G$ is a subset of edges which is both dominating and an
induced matching. Not every graph admits a DIM, and the problem of
determining whether a graph admits it is also known in  the
literature as \emph{efficient edge  domination problem}. The
weighted version of DIM problem is to find a DIM such that the sum
of weights of its edges is minimum among all DIM's, if any.
The counting version of the problem consists on counting the amount
of DIM's a graph has. It is easy to see that the weighted version
and the counting version of the problem are harder than the unweighted one.
If the graph $G$ has negative weights the problem can be solved using
the same algorithm that solves the problem for non-negative weights.
Let $-M$ be the minimum weight among all edges of $G$, modify the weights
of $G$ by adding $M$ to the weights of all edges.
%Since the number of edges of any DIM in $G$ is the same,
It is not hard to see that every DIM is a maximum induced matching, and hence the number of edges of every DIM in $G$ is the same.
Therefore the optimal solution for the modified graph is the same that the optimal solution for the original graph.

We assume the graph $G$ to be connected, otherwise, the DIM of $G$ is
the union of the DIM's of its connected component, and so we can
restrict to the  connected case.

We will use an alternative definition \cite{Do-Lo} of the problem
of finding a dominating induced matching. It asks to determine if
the vertex set of a graph $G$ admits a partition into two subsets.
The vertices of the first subset are called {\it white} and induce
an independent set of the graph, while those of the second subset
are named {\it black} and induce an 1-regular graph.

A straightforward brute-force algorithm for finding the DIM of a graph $G$
consists in finding all bipartitions of $V(G)$, color one of the parts as {\it white}, the
other as {\it black}, and checking if the result is a valid DIM. The complexity of
this algorithm is $O(2^n \cdot m) \in O^*(2^n)$.

\section{Extensions of Colorings}\label{sec-extensions}

Assigning one of the two possible colors, white or black, to vertices of $G$ is called a coloring of $G$. A coloring is
\emph{partial} if only part of the vertices of G have been assigned colors, otherwise it is \emph{total}. A black vertex is
{\it single} if it has no black neighbor, and is paired if it has exactly one black neighbor. Each coloring, partial or total, can
be {\it valid} or {\it invalid}.

Next, we describe the natural  conditions for determining if a coloring is valid or invalid.

\begin{definition}:\label{lemma-validation} RULES FOR VALIDATING COLORINGS: \\
The following are necessary and sufficient conditions for a
coloring to be valid: \bigskip

A partial coloring is valid whenever:
\begin{description}
\item [ V1.\label{itm:validI}] No two white vertices are adjacent, and 
\item [ V2.\label{itm:validII}] Each black vertex is either single or paired. 
Each single vertex has some uncolored neighbor.
\end{description}
\bigskip

A total coloring is valid whenever:
\begin{description}
\item [ V3,\label{itm:valid1} ] No two white vertices are adjacent, and 
\item [ V4.\label{itm:valid2} ]Each black vertex is paired.
\end{description}
\end{definition}

\begin{lemma}\label{cor-totalvalid}
There is a one-to-one correspondence between total valid colorings
and dominating induced matchings of a graph.
\end{lemma}

{\it Proof:} It follows from the definitions. $_\triangle$

For a coloring $C$ of the vertices of $G$, denote by
$C^{-1}(white)$ and $C^{-1}(black)$, the subsets of vertices
colored white and black. A coloring $C'$ is an {\it extension} of
a $C$ if  $C^{-1}(black)\subseteq C'^{-1}(black)$ and
$C^{-1}(white) \subseteq C'^{-1}(white)$. For $V', V'' \subset
V(G)$ if $C'$ is obtained from $C$ by adding to it the vertices of
$V'$ with the color black and those of $V''$ with the color white
then write $C = C' \cup BLACK(V') \cup WHITE(V'')$.
Note that a {\it total valid} coloring can be only an extension of 
{\it partial valid} colorings and itself.

Given a partial coloring $C$, the basic idea of the algorithm is
to iteratively find extensions $C'$ of $C$, until eventually a
total valid coloring is reached. It follows from the validation
rules that if $C$ is invalid, so is $C'$. Therefore, the algorithm
keeps checking for validation, and would discard an extension
whenever it becomes invalid.

Basically, there are two different ways of possibly extending a
coloring. First, there are partial colorings $C$ which force the
colors of some of the so far uncolored vertices, leading to an
extension $C'$ of $C$. In this case, say that $C'$ has been
obtained from $C$ by {\it propagation}. The following is a
convenient set of rules, whose application may extend $C$, in the
above described way.

\begin{lemma}:\label{lemma-propagation} RULES FOR PROPAGATING COLORS: \\

The following are forced colorings for the  extensions of a partial coloring of $G$.

\begin{description}
    \item [P1.\label{itm:ruleDiamond}] In an induced diamond, degree-3 vertices must be black and the remaining ones must be white

 		\item [P2.\label{itm:rulePendant}] The neighbor of a pendant vertex must be black
 
 		\item [P3.\label{itm:ruleWhiteN}] Each neighbor of a white vertex must be black
 
 		\item [P4.\label{itm:rulePaired}] Except for its pair, the neighbors of a paired (black) vertex
  				must be white
  				
  	\item [P5.\label{itm:ruleTwoB}] Each vertex with two black neighbors must be white
  
  	\item [P6.\label{itm:ruleOneN}] If a single black vertex has exactly one uncolored neighbor then this neighbor must be black
  
  	\item [P7.\label{itm:rulePaw}] In an induced paw, the two odd-degree vertices must have different colors
  	
  	\item [P8.\label{itm:ruleC4}] In an induced $C_4$, adjacent vertices must have different 
  				colors
  
  	\item [P9.\label{itm:ruleContained}] If the neighborhood of any  uncolored neighbor of a
    single (black) vertex $s$ is contained in the neighborhood of $s$
    then the uncolored neighbor $v$ of $s$
    minimizing weight(sv) must be black. If there are several options for vertex $v$, 
    choose any one of them. We  require rules P1 and P8 %\ref{itm:ruleDiamond} and \ref{itm:ruleC4} 
    to be applied before P9.%\ref{itm:ruleContained}.

\end{description}

\end{lemma}

\begin{proof}
The rules P1, P7, P8 follows from \cite{Br-Hu-Ne}. 
while rules P3, P4, P5, P6 follows from \cite{Do-Lo}.

%The rules \ref{itm:ruleDiamond}, \ref{itm:rulePaw}, \ref{itm:ruleC4} follows from \cite{Br-Hu-Ne}. 
%while rules \ref{itm:ruleWhiteN}, \ref{itm:rulePaired}, \ref{itm:ruleTwoB}, \ref{itm:ruleOneN} follows from \cite{Do-Lo}.
The rule P2 %\ref{itm:rulePendant} 
follows from the coloring definition since each black vertex must be paired in order for
the coloring to be {\it valid}.
Finally, for P9, %\ref{itm:ruleContained}, 
let $s$ be a single vertex. Suppose the neighborhood of any uncolored neighbor
of $s$ lies within the neighborhood of $s$. Then the choice of the
vertex to become the pair of $s$ is independent of the choices for
the remaining single vertices of the graph. Therefore, to obtain a
minimum weighted dominating induced matching of $G$, the neighbor
$v$ of $s$ minimizing $weight(sv)$ must be black. $_\triangle$
\end{proof}

\begin{lemma}\label{lem-K4}
\cite{Br-Hu-Ne} If $G$ contains a $K_4$ then $G$ has no DIM.
\end{lemma}

Say that a coloring $C$ is {\it empty} if all vertices are uncolored.

Let $C$ be a valid coloring and $C'$ an extension of it, obtained
by the application of the propagation rules.  If $C = C'$ then $C$
is called {\it stable}. On the other hand, if $C \neq C'$ then
$C'$ is not necessarily valid. Therefore, after applying
iteratively the propagation rules, we reach an extension which is
either invalid or stable.

In order to possibly extend a stable coloring $C$, we apply {\it
bifurcation rules}. Any coloring directly obtained by these rules
is not forced. Instead, in each of the these rules, there are two
possibly conflicting alternatives leading to distinct extensions
$C'_1, C'_2$ of $C$. Each of $C'_1$ or $C'_2$ may be independently
valid or invalid. The next lemma describes the bifurcation rules.
We remark that there exist simpler bifurcation rules. However,
using the  rules below we obtain a sufficient  number of vertices
that get forced colorings, through the propagation which follow
the application of any bifurcation rule, so as to guarantee a
decrease of the overall complexity of the algorithm. The
complexity obtained relies heavily on this fact.

In general, we adopt the following notation. If $C$ is  a stable
coloring then S denotes the set of single vertices of it
, $U$ is the set of uncolored vertices and $T = U \setminus \cup_{s \in S} N_U(s)$.

\begin{lemma}\label{lemma-bifurcation}: BIFURCATION RULES \\
Let $C$ be a partial (valid) stable coloring of a graph $G$. At least one of the
following alternatives can be applied to define extensions $C'_1,
C'_2$ of $C$.

\begin{description}

\item [B1.\label{itm:ruleEmpty}] If $C$ is an {\it empty} coloring: choose an arbitrary vertex $v$ then  $C'_1 := C \cup
BLACK(\{v\})$ and $C'_2 := C \cup WHITE(\{v\})$\\

\item [B2.\label{itm:ruleExEdge}] If $\exists$ edge $vw$ s.t. $v \in N_U(s)$ and $w \in
N_U(s')$, for some $s,s' \in S, s \ne s'$ then  $C'_1 := C \cup
BLACK(\{v\})$ and $C'_2 := C \cup WHITE(\{v\})$\\

\item [B3.\label{itm:ruleSomes}] For some $s \in S$, if $\exists v \in N_U(s)$  s.t. $\exists w \in N_T(v)$:

	\begin{description}
			\item [(a)\label{itm:Bif3a}] If $|N_U(s)| \neq 3 \vee d(w) \neq 3 \vee |N_T(v)| \geq 2$ then $C'_1 := C \cup BLACK(\{v\})$ and $C'_2 := C \cup WHITE(\{v\})$.	\\
			
			\item [(b)\label{itm:Bif3b}] If $|N_U(s)| = 3 \wedge d(w) = 3 \wedge N_T(v) = \{w\}$, let $N_U(s) = \{v,v',v''\}$. \\
				\begin{description}
				
					\item [i.] If $N_U(v') = N_U(v'') = \emptyset$ then $C'_1 := C \cup BLACK(\{v\})$ and
								$C'_2 := C \cup WHITE(\{v\})$\\
					
					\item [ii.] If $N_U(v') \neq \emptyset$, let $w' \in N_T(v')$, with $w' \neq w$. If $|N(w)\cup N(w')| > 5$ or $ww' \notin E(G)$ then $C'_1 := C \cup
								BLACK(\{v\})$ and $C'_2 := C \cup WHITE(\{v\})$ \\
					
					\item [iii.\label{itm:B3b.iii}] If $N_U(v') \neq \emptyset$, let  $w' \in N_T(v')$, with $w' \neq w$. 
								If  $ww'  \in E(G)$ and $z \in N(w) \cap N(w')$ then $C'_1 := C \cup BLACK(\{v''\})$, while if 
								$weight(sv) + weight(w'z) \leq weight(sv') + weight(wz)$ then $C'_2 := C \cup BLACK(\{v\})$, 
								otherwise $C'_2 := C \cup BLACK(\{v'\})$ \\
				\end{description}
	\end{description}
	
\end{description}
\end{lemma}

Each rule is applied after the previous rule, that is, if the
condition of the previous case is not verified in the entire
graph. Note that this applies to subitems of case B3.
%\ref{itm:ruleSomes}.

{\it Proof} If $C$ is an {\it empty} coloring then the rule B1 is applied.\\
		 If $C$ is not an {\it empty} coloring and $C$ is not a {\it total} coloring then $S \ne \emptyset$. 
		 Since $C$ is not {\it total} and the graph is connected then there is at least one edge $sv$
		 where $v$ is uncolored. If $s$ is white then $v$ must be black P3%\ref{itm:ruleWhiteN},
		  else if $s$ is a {\it paired} vertex then
		 $v$ must be white P4%\ref{itm:rulePaired}
		 . Therefore $s$ must be a single black vertex, hence $S \ne \emptyset$.
		 Let $s \in S$. Since $C$ is valid then $N_U(s) \ne \emptyset$ by %\ref{itm:validII} 
		 V2 and since is {\it stable} $|N_U(s)| \ne 1$ by P6%\ref{itm:ruleOneN}.
		 Therefore $|N_U(s)| \geq 2$.
		 Moreover rule P9%\ref{itm:ruleContained} 
		 can not be applied, therefore $\exists v \in N_U(s)$ s.t. 
		 $|N_U(v) \setminus N(s)| > 0$, let $w \in N_U(v) \setminus N(s)$. 
		 If $\exists s' \in S, s \ne s'$ s.t. $w \in N_U(s')$ then rule B2 %\ref{itm:ruleExEdge} 
		 is applied.\\
		 Suppose that rule P2%\ref{itm:rulePendant} 
		 can not be applied. 
		 Then  $w \in N_T(v) (|N_T(v)| \ge 1)$. Clearly, $d(w) \ne 1$, otherwise, rule P2%\ref{itm:rulePendant} 
		 must be applied and $v$ must get color black.
		 
		 In case $|N_U(s)| \ne 3$ or $d(w) \ne 3$ or $|N_T(v)| \geq 2$ we apply rule B3(a).
		 Otherwise: $|N_U(s)|=3,d(w)=3,|N_T(v)|=1$. Note that in B3(b) whenever we refer to $v'w'$ it behaves symmetric to $vw$ since otherwise $v'w'$
		 were found in step B3(a) replacing $vw$.\\
		 In the first subcase of B3(a) the case analyzed is whenever $N_U(v') = N_U(v'') = \emptyset$, while in the second and third 
		 the algorithm handle the cases when at least one of them has uncolored neighbors.\\
		 Suppose w.l.o.g. $N_U(v') \ne \emptyset$ where $w' \in N_T(v')$. It is easy to see that $w \ne w'$ since otherwise $svwv'$ is a $C_4$ and therefore $w$
		 can't be uncolored by rule P8.%\ref{itm:ruleC4}. \\
		 Now there are three cases which lead to two possible outcomes from the algorithm: In case $ww' \in E(G)$ or $|N_U(v) \cup N_U(w)|>5$ then the result of the 					 algorithm is given by the second subcase (ii), else it is given by the third subcase (iii). $_\triangle$
		 
		 %The case when $|N_U(s)|=3,d(w)=3,|N_T(v)|=1$ is solved in \ref{itm:Bif3b} while other cases in \ref{itm:Bif3a}.

\section{The Algorithm}\label{sec-algorithm}

The lemmas described in the last section lead to an exact
algorithm for finding a minimum weight DIM of a graph $G$, 
if any, which we describe below.

In the initial step of the algorithm, we find the set containing the $K_4$'s of $G$.
If  $K4 \ne \emptyset$, by lemma \ref{lem-K4}, $G$ does not
have DIM's, and terminate the algorithm. Otherwise, define the set
$COLORINGS$ to contain through the process the candidates
colorings to be examined and eventually extended. This set should be implemented using a $LIFO$ 
(Last In First Out) data structure which achieves linear space complexity of the algorithm because 
the number of colorings in $COLORINGS$ is at most $n+1$, and each coloring needs $O(1)$ space.
We give more detailed explanation in the next section.
Next, include in COLORINGS an {\it empty} coloring.
In the general step, we choose any coloring $C$ from $COLORINGS$
and remove it from this set. Then iteratively  propagate the
coloring by Lemma \ref{lemma-propagation} into an extension $C'$
of it, and validate the extension by Lemma \ref{lemma-validation}.
The iterations are repeated until one of the following situations
is reached: $C'$ is invalid, $C'$ is a total valid coloring, or a
partial stable (valid) coloring. In the first alternative, $C'$
is discarded and a new coloring from $COLORINGS$ is chosen. If
$C'$ is a a total valid coloring, then sum the amount of valid DIMs 
related to this coloring, find its weight and if
smaller than the least weight so far obtained, it becomes the
current candidate for the minimum weight of a DIM if $G$. Finally,
when $C'$ is stable  we extended it by bifurcation rules: choose
the first rule of Lemma \ref{lemma-bifurcation} satisfying $C'$,
compute the extensions $C'$ and $C''$, insert them in $COLORINGS$,
select a new coloring from $COLORINGS$ and repeat the process.

The formulation below describes the details of the method. The
propagation and validation of a coloring $C$ is done by the
procedure $PROPAGATE-VALIDATE(C, RESULT)$. 
At the end, the returned coloring corresponds to the extension $C'$ of $C$, 
after iteratively applying propagation. The variable RESULT indicates
the outcome of the validation analysis. If $C'$ is invalid then
$RESULT$ is returned as `invalid'; if $C'$ is a valid total
coloring then it contains `total', and otherwise $RESULT$ equals
`partial'. Finally, $BIFURCATE(C, C'_1,C'_2)$ computes the
extensions $C'_1$ and $C'_2$ of $C$.  \\

\newpage

{\bf Algorithm Minimum Weighted DIM / Counting DIM}
\vspace{1mm}

\fbox{%
\begin{minipage}{5 in}

\begin{description}
	\item [1.] Find the subset $K4$\\
						 {\bf if} $K4 \neq \emptyset$ {\bf then} terminate the algorithm: $G$ has no DIM \\
						 \hspace*{.4cm} $SOLUTION := NO DIM$ 
	
	\item [2.] $COLORINGS := \{C\}$, where $C$ is an {\it empty} coloring

	\item [3.] {\bf while} $COLORINGS \neq \emptyset$ {\bf do} 
	
				\begin{description}
						\item [a.] \hspace*{.3cm} choose $C \in COLORINGS$ and remove it from $COLORINGS$
						\item [b.] \hspace*{.3cm} $PROPAGATE-VALIDATE(C, RESULT)$
						\item [c.] \hspace*{.3cm} {\bf if} $RESULT =$ `total' {\bf and} $weight(C) < SOLUTION$ {\bf then}
															\hspace*{.9cm} $SOLUTION := weight(C)$\\
												\hspace*{.3cm} {\bf else if} $RESULT =$ `partial' {\bf then} \\
														\hspace*{.6cm} Set $C'_1$ and $C'_2$ according to BIFURCATION RULES on $C$ B1\\%\ref{itm:ruleEmpty}\\
														\hspace*{.6cm} $COLORINGS := COLORINGS \cup \{C'_1, C'_2\}$\\
												\hspace*{.3cm} {\bf end if} 
				\end{description}
								
	\item [4.] {\bf Output} $SOLUTION$
\end{description}
\end{minipage}}

\vspace{1cm}

{\bf procedure} $PROPAGATE-VALIDATE (C, RESULT)$
\vspace{1mm}

\fbox{%
\begin{minipage}{5 in}

\begin{description}

	\item [Comment] Phase 1: Propagation

	\item [1.] $C':= C$
	
	\item [2.] {\bf repeat} \\
						\ \ \ \ $C := C'$ \\
						\ \ \ \ $C' :=$ extension of $C$ obtained by the PROPAGATION RULES \\ %\ref{lemma-propagation}  \\
						{\bf until} $C = C'$ 
						
	\item [Comment] Phase 2: Validation
	
	\item [3.] Using the VALIDATION RULES (lemma \ref{lemma-validation}) do as follows: \\
						{\bf if} $C$ is an invalid coloring {\bf then return} $(C, `invalid')$ \\
						{\bf else} if $C$ is a partial coloring {\bf then return} $(C, `partial')$ \\
						{\bf else return} $(C,$ `total'$)$ \\

\end{description}

\end{minipage}}

\section{Correctness and Complexity}\label{sec-correctness-complexity}

It is easy to see that our algorithm uses the {\it branch $\&$ reduce} paradigm since {\it propagation rules} can be mapped to {\it reduction rules} since are used
to simplify the problem instance or halt the algorithm and the {\it bifurcation rules} can be mapped to {\it branching rules} since are used to solve the problem
instance by recursively solving smaller instances of the problem.

\begin{theorem}\label{thm-correctness}
The algorithm described in the previous section correctly computes
the minimum weight of a dominating induced matching of a graph
$G$.
\end{theorem}

{\it Proof}: 
The correctness of the algorithm follows from the fact that the algorithm considers all the cases that need to be considered, 
this is, any coloring that represents a DIM must be explored. 
Lemmas \ref{lemma-propagation} and \ref{lemma-bifurcation} ensures that the simplifications of the instances are valid, {\it invalid} colorings are 
discarded, some valid colorings can not be explored only if other valid coloring representing a better DIM (with less weight) is explored.

For proving the worst-case running time upperbound for the algorithm we will use the following useful definition and theorem.

\begin{definition}\label{def:Branch}
\cite{Fo-Kr}
Let $b$ a branching rule and $n$ the size of the instance. Suppose rule $b$ branches the current instance into $r \geq 2$ instances of size at most $n-t_1,n-t_2, \ldots, n-t_r$, for all instances of size $n \geq max \{t_i:i = 1, 2, \ldots, r\}$. Then we call $b = (t_1,t_2,\ldots, t_r)$ the {\it branching vector} of branching
rule $b$.
\end{definition}

The branching vector $b = (t_1, t_2, \ldots, t_r)$ implies the linear recurrence $T(n) \leq T(n-t_1) + T(n-t_2) + \ldots, T(n-t_r)$.

\begin{theorem}\label{thm:Branch}
\cite{Fo-Kr}
Let $b$ be a branching rule with branching vector $(t_1, t_2, \ldots ,t_r)$. Then
the running time of the branching algorithm using only branching rule $b$ is $O^*(\alpha^n)$,
where $\alpha$ is the unique positive real root of
\begin{center}
	$x^n - x^{n-t_1} - x^{n-t_2} - \ldots - x^{n-t_r} = 0$
\end{center}
\end{theorem}
 
The unique positive real root $\alpha$ is the {\it branching factor} of the branching vector $b$. 

We denote the branching factor of $(t_1, t_2, \ldots, t_r)$ by $\tau(t_1, t_2, \ldots, t_r)$.

Therefore for analyzing the running time of a branching algorithm we can compute the factor $\alpha_i$ for every branch rule $b_i$, and an upper bound of the 
running time of the branching algorithm is obtained by taking $\alpha = max_i \alpha_i$ and the result is an upper bound for the running time of $O^*(\alpha^n)$.

The upper bound comes from counting the leaves of the search tree given by the algorithm, using the fact that each leave can be executed in polynomial time.
The complexity of the algorithm without hiding the polynomial depends on the upperbound time for the execution of each branch in the search tree.

Further notes on this topic can be found in \cite{Fo-Kr}

\begin{theorem}\label{thm-complexity}
The algorithm above described requires $O^*(1.1939^n)$ time and $O(n+m)$ space for completion.
\end{theorem}

{\it Proof}: 

Using the definition \ref{def:Branch} and the theorem \ref{thm:Branch} the calculation of the upper bound time is reduced to calculation
of the {\it branching vector} for each branching rule (i.e. bifurcation rules in our algorithm) and obtain the associated {\it branching factor} for each case.
Then the bound is given by the maximum {\it branching factor}.
Note that to use this we must observe that the {\it reduction rules} (i.e. propagation rules in our algorithm) can be computed in polynomial time and
leads to at most one valid extension of the considered coloring. So, the propagation rules do not affect the exponential factor of the algorithm.
Moreover, each branch of the algorithm has cost $O(n+m)$ in time and space. 
This is easy to note since from the {\it empty} coloring up to any {\it total} coloring each vertex $v$ is painted once and the cost in time incurred 
for painting each vertex is given by the updating of the color of the vertex and updating this information for the neighborhood, hence $
|N(v)|$ times a constant operation for updating a counter with amount of black/white/uncolored neighbors. Therefore, the total cost for each branch is $O(n+m)$.

Let's analyse each bifurcation rule to obtain the maximum {\it branching factor}:

\begin{description}

\item [1.] If $C$ is an {\it empty} coloring: choose an arbitrary vertex $v$ then  $C'_1 := C \cup BLACK(\{v\})$ and $C'_2 := C \cup WHITE(\{v\})$:
			It is easy to see that this rule is executed once, after that, the coloring is never empty again. Since this rule bifurcation opens two branchs
			then we can upper bound the time of the algorithm by $2$ times the complexity of the algorithm executed in an instance of size $n-1$. 
			Therefore the asymptotic behavior of the algorithm is not affected.

\item [2.]	{\it If $\exists$ edge $vw$ s.t. $v \in N_U(s)$ and $v' \in N_U(s')$, for some $s,s' \in S, s \ne s'$ then  
					$C'_1 := C \cup BLACK(\{v\})$ and $C'_2 := C \cup WHITE(\{v\})$.}

					Here we extend the original coloring $C'$ to $C'_1$ and $C'_2$ by coloring the vertex $v$ with black and white respectively. 
					Recall that exists an edge $vw$ such that $v \in N_U(s)$, $w \in N_U(s')$. 
					If $v$ is black then $N_U(s) \setminus v$ are white, while $v'$ is white. 
					On the other hand, if $v$ is white then $w$ is black and $N_U(s') \setminus w$ are white. 
					Therefore the size of uncolored vertices is reduced for each branch (i.e. for each new coloring). 
					The associated branching vector is $(1 + |N_U(s)|, 1 + |N_U(s')|)$.  
					By rule P2 %\ref{itm:rulePendant} 
					$|N_U(s)| \geq 2$ and $|N_U(s')| \geq 2$. 
					The following observation turns out to be useful: 
					If $|N_U(s_i)| = 2$ then $N_U(s_i)$ can be totally painted wether $v$ is black or white. 
					The case with $N_u(s')$ is symmetric. Therefore the branching vector with biggest branching factor is (3,5) $(\tau(3,5)\approx 1.1939)$.

\item [3.] For some $s \in S$, if $\exists v \in N_U(s)$  s.t. $\exists w \in N_T(v)$:

					Note that if $\not \exists w \in N_T(v)$ for any $v \in N_U(s)$ then either the propagating rule P9%\ref{itm:ruleContained} 
					or P5%\ref{itm:ruleTwoB}
					can be applied to get an extension of the coloring.

	\begin{description}
			\item [(a)] If $|N_U(s)| \neq 3 \vee d(w) \neq 3 \vee |N_T(v)| \geq 2$ then $C'_1 := C \cup BLACK(\{v\})$ and $C'_2 := C \cup WHITE(\{v\})$.	\\
			
			Since $v$ is uncolored then $w$ is not a {\it pendant} vertex, $d(w)>1$.
			Since $w$ is uncolored then it has nor white nor paired black neighbor. Moreover, if $w$ has a single black neighbor
			then this is the case analyzed above. Therefore $w$ has uncolored neighbors and let $x$ be one of them. \\
			
			\begin{description}
				\item [(a.1)] $|N_T(v)| \geq 2$: Let $v' \in N_T(v)$. 
							In $C'_1$ $\{v,x\}$ will be black while $\{v_1,v',w\}$ will be white. 
							In $C'_2$ $\{v\}$ will be white while $\{v',w\}$ will be black. This lead to the branching vector (3,5).\\
										
				\item [(a.2)] $d(w) \neq 3$. 
							If $d(w)=2$ then in $C'_1$ the vertices $N_U(s) \cup \{w,x\}$ will be colored and in $C'_2$ 
							the vertices $\{v,x\}$ will be black while $\{w\}$ will be white. Therefore the branching vector will be
							at least (3,5).\\
							Else if $d(w) > 3$ then in $C'_1$ the vertices $N_U(s) \cup N_U[w]$ will be colored and in $C'_2$
							the vertices $\{v,w\}$ will be colored. In case $|N_U(s)|=2$ then $v_1$ will be colored too.
							Therefore the branching vector (2,7) $(\tau(2,7)=1.1908)$.\\

				\item [(a.3)] $|N_U(s)|=2$: Let $N_U(s) = \{v,v_1\}$ and $N(w) = \{v,x,x'\}$. 
							In $C'_1$ after applying propagation rules the vertices $\{v,x,x'\}$ will be black while $\{v_1,w\}$ will be white.
							In $C'_2$ after applying propagation rules the vertices $\{v_1,w\}$ will be black while $\{v\}$ will be white.
							The result is the branching vector (3,5).\\
							
				\item [(a.4)] $|N_U(s)|>3$: Let $\{v_1,v_2,v_3\} \in N_U(s)$ and $N(w) = \{v,x,x'\}$. 
						In $C'_1$ after applying propagation rules the vertices $\{v,x,x'\}$ will be black while $\{v_1,v_2,v_3,w\}$ will be white.
						In $C'_2$ after applying propagation rules the vertices $\{w\}$ will be black while $\{v\}$ will be white.
						The result is the branching vector (2,7)\\		

			\end{description}

			\item [(b)] {\it If $|N_U(s)| = 3 \ \wedge \ d(w) = 3$ where $N_U(w)=\{v,x,x'\}, N_U(s) = \{v,v',v''\}$}
			 Note that $\{x,x'\} \cap \{v,v',v''\} = \emptyset$ since otherwise at least one of them must be colored by rule P8.%\ref{itm:ruleC4}.
			 
			 \begin{description}
				
					\item [(b.1)]{\it If $N_U(v') = N_U(v'') = \emptyset$ then \\
								\ \ \ \ $C'_1 := C \cup BLACK(\{v\})$ and $C'_2 := C \cup WHITE(\{v\})$ }:\\ 
								Suppose w.l.o.g. $weight(sv') \leq weight (sv'')$, then \\
								In $C'_1$ after applying propagation rules the vertices $\{v,w',w''\} $ will be black while $\{v',v'',w\} $ will be white.
								In $C'_2$ after applying propagation rules the vertices $\{v',w\} $ will be black while $\{v,v''\} $ will be white.
								The result is the branching vector (4,6) $(\tau(4,6)=1.1510)$.\\

					\item [(b.2)] {\it If $N_U(v') \neq \emptyset$, let $w' \in N_T(v')$, with $w' \neq w$: \\
							  If $|N_T[w] \cup N_T[w']| > 5$} then \\
								\ \ \ \ $C'_1 := C \cup BLACK(\{v\})$ and $C'_2 := C \cup WHITE(\{v\})$\\
								
								Note that if $d(w') \neq 3$ then $v'w'$ satisfies the properties of an already analized case, hence $C'_1 := C \cup BLACK(\{v'\})$
								and $C'_2 := C \cup WHITE(\{v'\})$. 
								
								Since $d(w) = d(w') = 3$ and $|N_T[w] \cup N_T[w']| > 5$, then $\exists x,y$ s.t. $x \in N_T(w), x \notin N_T(w')$ and
								$y \in N_T(w'), y \notin N_T(w)$.
								In $C'_1$ after applying propagation rules the vertices $\{v,x,x',w'\}$ will be black while $ \{v',v'',w\}$ will be white. If $x'=w'$ then $y$ must
								be black by rule P6.%\ref{itm:ruleOneN}.
								In $C'_2$ the vertex $\{w\}$ will be black while the vertex $\{v\}$ will be white.
								The result is the branching vector (2,7)\\

					\item [(b.3)] {\it If $N_U(v') \neq \emptyset$, let $w' \in N_T(v')$, $w' \neq w$\\ 
								If $|N_T[w] \cup N_T[w']| \leq 3$ and $z \in N(w) \cap N(w')$} then \\
								\ \ \ \ $C'_1 := C \cup BLACK(\{v''\})$, \\
								if  $weight(sv) + weight(w'z) \leq weight(sv') + weight(wz)$ then \\
								\ \ \ \ $C'_2 := C \cup BLACK(\{v\})$\\
								otherwise $C'_2 := C \cup BLACK(\{v'\})$ \\
								
								Since $d(w) = d(w') = 3$ then $ww' \in E(G)$ and $\exists z \in N_T(v) \cap N_T(w)$, otherwise the case is one of the above.
					
								In both colorings, $C'_1$ and $C'_2$ the vertices $\{v,v',v'',w,w',z\}$ will be colored. The branching vector is (6,6).
								$(\tau(6,6)=1.1225)$.
	
				\end{description}
				
	\end{description}
\end{description}

The worst branching factor is  $\tau(3,5)\approx 1.1939$. In consequence, the time complexity of this algorithm is $O*(1.1939^n)$.

To achieve linear space complexity, we use a stack to store the coloring sequence of the current branch. 
The only additional space is needed for $COLORINGS$ and extra information to restore the initial condition for each coloring. 
For each coloring $c \in COLORINGS$ extended from a bifurcation rule we store the number of colored vertices before the bifurcation, 
the vertex colored during bifurcation and its color. These elements are sufficient to restore the initial condition.
$_\triangle$

%Note that for the case of non-connected graphs, we have $m<n$
The analysis can be extended for the case of non-connected graphs. It is easy to obtain the same upper bound after separating 
the cases where each connected component of four or less vertices is solved in constant time.

\section{Counting the number of DIM's}

The previous algorithm can be easily adapted to count the number of DIM's. The number of DIM's is the number of {\it total valid} colorings.
Given a coloring $C$ we define $TVC(C)$ the number of {\it total valid} colorings that can be extended from $C$.
If we apply any propagation rule to coloring $C$ we obtain a coloring $C'$. Clearly $TVC(C) = TVC(C')$, except for rule P9%\ref{itm:ruleContained}
. In the
later case $TVC(C) = TVC(C') \cdot |N_U(s)|$ where $s$ is the single vertex chosen to apply the rule.

If we apply any bifurcation rule to coloring $C$ we obtain two extended colorings $C'_1$ and $C'_2$. Clearly $TVC(C) = TVC(C'_1) + TVC(C'_2)$, except for rule B3(b)iii.
In the later case $TVC(C) = TVC(C'_1) + 2 \cdot TVC(C'_2)$.

Using the above facts it is trivial to modify the algorithm to solve the counting problem.

\section{Conclusions}

We have developed a new exact exponential algorithm for an extensively studied problem. Moreover the developed algorithm is practical since there are no big constants
or polynomials hidden in the upper-bound and it is straightforward to implement it. 
\bigskip

\begin{tabular}{|l | l | l |}
\hline
Problem & Previous results & New results\\
\hline
Weighted DIM & $O^*(1.4423^m)$ \cite{Br-Le-Ra,Milanic,Ts}, $O(1.4423^n \cdot m)$ \cite{DBLPMin} & $O(1.1939^n \cdot m)$ \\
\hline
Counting DIM & $O^*(1.3247^m)$ \cite{Br-Le-Ra,Milanic,Dahllof2002}, $O(1.4423^n \cdot m)$ \cite{DBLPMin} & $O(1.1939^n \cdot m)$\\
\hline
\end{tabular}

\end{document}